\documentclass[a4paper, 12pt]{article}

\usepackage[sort&compress]{natbib}
\bibpunct{(}{)}{;}{a}{}{,} 

\usepackage{amsthm, amsmath, amssymb, mathrsfs, multirow, url, subfigure}
\usepackage{graphicx} 
\usepackage{ifthen} 
\usepackage{amsfonts}
\usepackage[usenames]{color}
\usepackage{fullpage}
\usepackage{algorithm, algorithmicx, algpseudocode}

\theoremstyle{plain} 
\newtheorem{thm}{Theorem}

\theoremstyle{definition}

\newcommand{\E}{\mathsf{E}}
\newcommand{\V}{\mathsf{V}}

\newcommand{\unif}{{\sf Unif}}
\newcommand{\nm}{{\sf N}}
\newcommand{\expo}{{\sf Exp}}
\newcommand{\gam}{{\sf Gamma}}

\renewcommand{\phi}{\varphi}

\title{Simulating from a gamma distribution with small shape parameter}
\author{
Chuanhai Liu \\
Department of Statistics \\
Purdue University \\
{\tt chuanhai@purdue.edu} \\
\mbox{} \\
Ryan Martin \quad and \quad Nick Syring \\
Department of Mathematics, Statistics, and Computer Science \\
University of Illinois at Chicago \\
{\tt (rgmartin, nsyring2)@uic.edu} 
}
\date{\today}

\begin{document}

\maketitle 

\begin{abstract}  
Simulating from a gamma distribution with small shape parameter is a challenging problem.  Towards an efficient method, we obtain a limiting distribution for a suitably normalized gamma distribution when the shape parameter tends to zero.  Then this limiting distribution provides insight to the construction of a new, simple, and highly efficient acceptance--rejection algorithm.  Comparisons based on acceptance rates show that the proposed procedure is more efficient than existing acceptance--rejection methods.    

\smallskip

\end{abstract}

\section{Introduction}
\label{S:intro}

Let $Y$ be a positive gamma distributed random variable with shape parameter $\alpha > 0$, denoted by $Y \sim \gam(\alpha,1)$.  The probability density function for $Y$ is given by 
\[ p_\alpha(y) = \frac{1}{\Gamma(\alpha)} y^{\alpha-1} e^{-y}, \quad y > 0, \]
where the normalizing constant, $\Gamma(\alpha) = \int_0^\infty y^{\alpha-1} e^{-y} \,dy$, is the gamma function evaluated at $\alpha$.  This is an important distribution in statistics and probability modeling.  In fact, since the gamma distribution is closely tied to so many important distributions, including normal, Poisson, exponential, chi-square, F, beta, and Dirichlet, one could argue that it is one of the most fundamental \citep[e.g.,][]{johnson.kotz.bala.book.1994}.  Here we are particularly interested in the problem of simulating gamma random variables when the shape parameter $\alpha$ is small.  This problem is important because the small-shape gamma, with a large scale parameter, is a simple but useful model for positive size or lifetime random variables with small mean but large variance \citep[e.g.,][]{kleiber.kotz.2003}.  Unfortunately, the small-shape gamma distribution is not easy to work with, so, seemingly routine calculations can become inefficient or even intractable.  For example, in R, the functions related to the gamma distribution---in particular, the {\tt rgamma} function for sampling---become relatively inaccurate when the shape parameter is small \citep[][``GammaDist'' documentation]{Rmanual}.  To circumvent these difficulties, and to move towards new and more efficient software, we show that $\gam(\alpha,1)$, suitably normalized, has a simple and non-degenerate limiting distribution as $\alpha \to 0$.  This result is then used to develop a new and efficient algorithm for sampling from a small-shape gamma distribution.  

When the shape parameter $\alpha$ is large, it follows from the infinite-divisibility of the gamma distribution and Lindeberg's central limit theorem \citep[][Sec.~27]{billingsley} that the distribution of $Y$ is approximately normal.  Specifically, as $\alpha \to \infty$, 
\[ \alpha^{-1/2}(Y - \alpha) \to \nm(0,1) \quad \text{in distribution}. \]
For large finite $\alpha$, better normal approximations can be obtained by working on different scales, such as $\log Y$ or $Y^{1/3}$.  Our interest is in the opposite extreme case, where the shape parameter $\alpha$ is small, approaching zero.  Here, neither infinite-divisibility nor the central limit theorem provide any help.  In Theorem~\ref{thm:limit} below, we prove that $-\alpha \log Y$ converges in distribution to $\expo(1)$, the unit-rate exponential distribution, as $\alpha \to 0$.   

Motivated by the limit distribution result in Theorem~\ref{thm:limit}, we turn to the problem of simulating from a small-shape gamma distribution.  This is a challenging problem with many proposed solutions; see, for example, \citet{best1983}, \citet{kundu.gupta.2007}, \citet{tanizaki2008}, and \citet{xi.tan.liu.2013}.  For small shape parameters, the default methods implemented in R and MATLAB, due to \citet{ahrens.dieter.1974} and \citet{marsaglia.tsang.2000}, respectively, have some shortcomings in terms of accuracy and/or efficiency.  The exponential limit in Theorem~\ref{thm:limit} for the normalized gamma distribution suggests a convenient and tight envelope function to be used in an acceptance--rejection sampler \citep[e.g.,][]{devroye1986, flury1990}.  We flesh out the details of this new algorithm in Section~\ref{S:simulation} and provide R code at \url{www.math.uic.edu/~rgmartin}.  This new method is simple and, as we demonstrate in Section~\ref{S:efficiency}, is more efficient than existing methods in terms of acceptance rates.    


\section{Limit distribution result}
\label{S:limit}

For the gamma function $\Gamma(z)$ defined above, write $f(z) = \log \Gamma(z)$.  Then the digamma and trigamma functions are defined as $f_1(z) = f'(z)$ and $f_2(z) = f''(z)$, the first and second derivatives of the log gamma function $f(z)$.  Recall that these are related to the mean and variance of $\log Y$, with $Y \sim \gam(\alpha,1)$:
\[ \E_\alpha (\log Y) = f_1(\alpha) \quad \text{and} \quad \V_\alpha (\log Y) = f_2(\alpha). \]
These formulae are most directly seen by applying those well-known formulae for means and variances in regular exponential families \citep[][Corollary~2.3]{brown1986}.  Next, write $Z=-\alpha \log Y$.  To get some intuition for why multiplication by $\alpha$ is the right normalization, consider the following recurrence relations for the digamma and trigamma functions \citep[][Chap.~6]{abramowitz.stegun.1966}:
\[ f_1(\alpha) = f_1(\alpha + 1) - 1/\alpha \quad \text{and} \quad f_2(\alpha) = f_2(\alpha + 1) + 1/\alpha^2. \]
Then, as $\alpha \to 0$,  
\begin{align*}
\E_\alpha(Z) & = -\alpha f_1(\alpha) = -\alpha f_1(\alpha+1) + 1 = O(1), \\
\V_\alpha(Z) & = \alpha^2 f_2(\alpha) = \alpha^2 f_2(\alpha + 1) + 1 = O(1).
\end{align*}
That is, multiplication by $\alpha$ stabilizes the first and second moments of $\log Y$.  Towards a formal look at the limiting distribution of $Z$, define the characteristic function 
\begin{equation}
\label{eq:character}
\phi_\alpha(t) = \E_\alpha(e^{it Z}) = \E_\alpha(Y^{-i\alpha t}) = \Gamma(\alpha - i\alpha t)/\Gamma(\alpha), 
\end{equation}
where $i=\sqrt{-1}$ is the complex unit.   

\begin{thm}
\label{thm:limit}
For $Y \sim \gam(\alpha,1)$, $-\alpha \log Y \to \expo(1)$ in distribution as $\alpha \to 0$.  
\end{thm}

\begin{proof}
Set $Z=-\alpha \log Y$.  The gamma function satisfies $\Gamma(z) = \Gamma(z+1)/z$, so the characteristic function $\phi_\alpha(t)$ for $Z$ in \eqref{eq:character} can be re-expressed as 
\[ \phi_\alpha(t) = \frac{\Gamma(\alpha - i \alpha t)}{\Gamma(\alpha)} = \frac{\Gamma(1 + \alpha - i \alpha t)/(\alpha - i \alpha t)}{\Gamma(1 + \alpha) / \alpha} = \frac{1}{1-it} \,\frac{\Gamma(1 + o_\alpha)}{\Gamma(1+o_\alpha)}, \]
where $o_\alpha$ are terms that vanish as $\alpha \to 0$.  Since the gamma function is continuous at 1, the limit of $\phi_\alpha(t)$ as $\alpha \to 0$ exists and is given by $1/(1-it)$.  This limit is exactly the characteristic function of $\expo(1)$, so the claim follows by L\'evy's continuity theorem.  
\end{proof}

\section{Small-shape gamma simulations}
\label{S:simulation}


Simulating from a gamma distribution with small shape parameter is a challenging problem that has attracted considerable attention in the literature; see, e.g., \citet{best1983}, \citet{kundu.gupta.2007}, \citet{tanizaki2008}, and \citet{xi.tan.liu.2013}.  Here we demonstrate that the limiting distribution result in Theorem~\ref{thm:limit} helps provide an improved algorithm for simulating gamma random variables with small shape parameter.  

For $Y \sim \gam(\alpha,1)$ with $\alpha$ near zero, let $Z=-\alpha \log Y$.  To simulate from the distribution of $Z$, one might consider an acceptance--rejection scheme; see, for example, \citet[][Chap.~20.4]{lange1999} or \citet[][Chap.~6.2.3]{givens.hoeting.2005}.  For this, one needs an envelope function that bounds the target density and, when properly normalized, corresponds to the density function of a distribution that is easy to simulate from.  By Theorem~\ref{thm:limit} we know that $Z$ is approximately $\expo(1)$ for $\alpha \approx 0$.  More precisely, the density $h_\alpha(z)$ of $Z$ has a shape like $e^{-z}$ for $z \geq 0$.  Therefore, we expect that a function proportional to an $\expo(1)$ density will provide a tight upper bound on $h_\alpha(z)$ for $z \geq 0$.  We shall similarly try to bound $h_\alpha(z)$ by an oppositely-oriented exponential-type density for $z < 0$, as is standard in such problems.  

The particular bounding envelope function $\eta_\alpha(z)$ is chosen to be as tight an upper bound as possible.  This is done by picking optimal points of tangency with $h_\alpha(z)$.  For this, we shall need a formula for $h_\alpha(z)$, up to norming constant, which is easily found:
\[ h_\alpha(z) = ce^{-z - e^{-z/\alpha}}, \quad z \in (-\infty, \infty). \]
The norming constant $c$ satisfies $c^{-1}=\Gamma(\alpha+1)$.  By following standard techniques, as described in \citet[][Chap.~20.4]{lange1999}, we obtain the optimal envelope function  
\[ \eta_\alpha(z) = \begin{cases} c e^{-z}, & \text{for $z \geq 0$}, \\ c w \lambda e^{\lambda z}, & \text{for $z < 0$}, \end{cases} \]
where $\lambda = \lambda(\alpha) = \alpha^{-1}-1$ and $w = w(\alpha) = \alpha / e(1-\alpha)$.  Plots of the (un-normalized) target density $h_\alpha(z)$ along with the optimal envelope $\eta_\alpha(z)$, for two small values of $\alpha$, are shown Figure~\ref{fig:envelope}.  The normalized envelope function $\eta_\alpha(z)$ corresponds to the density function of a mixture of two (oppositely-oriented) exponential distributions, i.e., 
\[ \frac{1}{1+w} \expo(1) + \frac{w}{1+w} \{-\expo(\lambda)\}, \]
which is easy to sample from using standard tools, such as {\tt runif} in R.  

\begin{figure}
\begin{center}
\subfigure[$\alpha=0.1$]{\scalebox{0.6}{\includegraphics{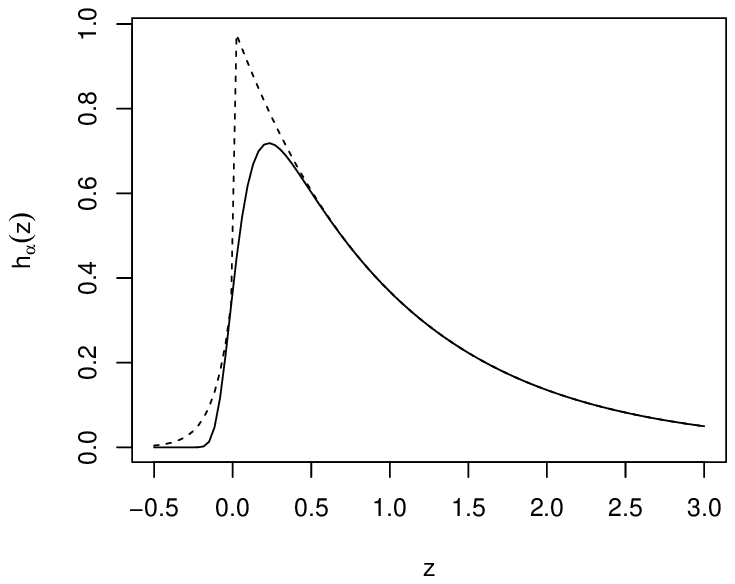}}}
\subfigure[$\alpha=0.05$]{\scalebox{0.6}{\includegraphics{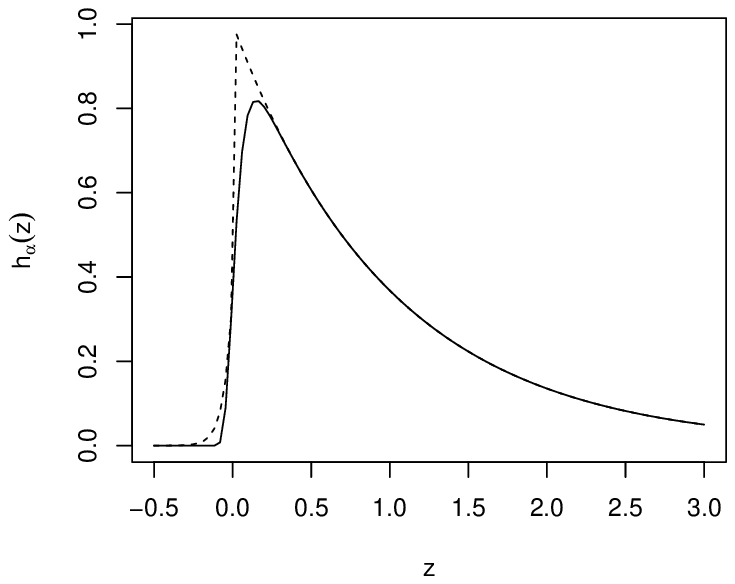}}}
\end{center}
\caption{Plots of the (un-normalized) target $h_\alpha(z)$ (solid) and envelope $\eta_\alpha(z)$ (dashed) for two values of $\alpha$.}
\label{fig:envelope}
\end{figure}

Pseudo-code for the proposed new program, named {\tt rgamss}, for simulating from a small-shape gamma distribution based on this acceptance--rejection scheme is presented in Algorithm~\ref{alg:code}.  R code is also available at \url{www.math.uic.edu/~rgmartin}, which provides the user with further options.  As a side note, since numerical precision can be lost in the final exponentiation step in Algorithm~\ref{alg:code}, we recommend returning the samples on the log-scale, which is the default in our R function.  


\begin{algorithm*}[t]
\begin{algorithmic}[1]
\State {\bf set} $\lambda \gets \lambda(\alpha)$, $w \gets w(\alpha)$, and $r \gets r(\alpha)$ as in the text.
\Loop
\State $U \gets \unif$ \hfill \#$\unif$ denotes the random number generator
\If {$U \leq r$} 
\State $z \gets -\log(U / r)$ 
\Else 
\State $z \gets \log(\unif) / \lambda$ 
\EndIf
\If {$h_\alpha(z) / \eta_\alpha(z) > \unif$} 
\State $Z \gets z$
\State {\bf break}
\EndIf 
\EndLoop
\State $Y \gets \exp(-Z / \alpha)$
\end{algorithmic}
\caption{-- Pseudo-code for the program {\tt rgamss} designed to simulate from a gamma distribution with small shape parameter $\alpha$.}
\label{alg:code}
\end{algorithm*}

\section{Efficiency comparisons}
\label{S:efficiency}

All methods based on the acceptance--rejection principle will provide genuine samples from the target distribution.  So, the most natural way to compare such methods is based on the acceptance rate.  It is common practice nowadays to employ an acceptance--rejection sampler inside a Markov chain Monte Carlo method, so having high acceptance rates results in shorter Monte Carlo run times.  

The acceptance rate $r(\alpha)$ for the proposed method is 
\begin{equation}
\label{eq:accept}
r(\alpha) = \{1+w(\alpha)\}^{-1} = \Bigl\{ 1 + \frac{\alpha}{e(1-\alpha)}\Bigr\}^{-1}. 
\end{equation}
It is clear from the approximation $r(\alpha) \approx 1-\alpha/e$ for $\alpha \approx 0$, that the acceptance rate converges to 1 as $\alpha \to 0$.  This indicates the high efficiency of the proposed method when $\alpha$ is small.  This is to be expected based on Theorem~\ref{thm:limit}: when $\alpha \approx 0$, $Z$ is approximately $\expo(1)$, so an algorithm that proposes $\expo(1)$ samples with probability $1/(1+w) \approx 1$ will likely accept the proposal.  To justify the claim in Section~\ref{S:intro} that the proposed method is more efficient than existing methods, all based on the accept--reject principle, we provide a comparison in terms of acceptance rates.  The methods being compared to {\tt rgamss} are Ahrens--Dieter \citep{ahrens.dieter.1974}, Best \citep{best1983}, and Algorithm~3 of Kundu--Gupta \citep{kundu.gupta.2007}; the methods of \citet{tanizaki2008} and \citet{xi.tan.liu.2013} were also considered, but these are not efficient enough to be compared with the others.  Plots of the acceptance rates for these methods, as a function of the shape parameter $\alpha$, are displayed in Figure~\ref{fig:accept}.  The proposed {\tt rgamss} has the highest acceptance rate over a range of small $\alpha$ values, namely $(0,0.3]$, and, therefore, is most efficient.  For $\alpha > 0.3$, the proposed method's efficiency drops below that of Kundu--Gupta but stays above Best and Ahrens--Dieter for $\alpha \in (0.3, 0.4]$ (not shown).  

There are additional numerical advantages to the proposed method beyond simulation efficiency.  Indeed, when $\alpha$ is very small, the gamma distribution is tightly concentrated around zero, so getting practically non-zero values can be challenging.  For example, when $\alpha=0.001$, nearly half of the samples returned by the Ahrens--Dieter method are exact zeros.  The Kundu--Gupta method has similar difficulties, since $U^{1/\alpha}$, for $U \sim \unif(0,1)$, is often practically zero.  The proposed method, on the other hand, actually works on the log-scale, a more appropriate scale for very small numbers, so one can readily obtain genuine non-degenerate samples of $\log Y$ for very small values of $\alpha$.

\begin{figure}
\begin{center}
\scalebox{0.8}{\includegraphics{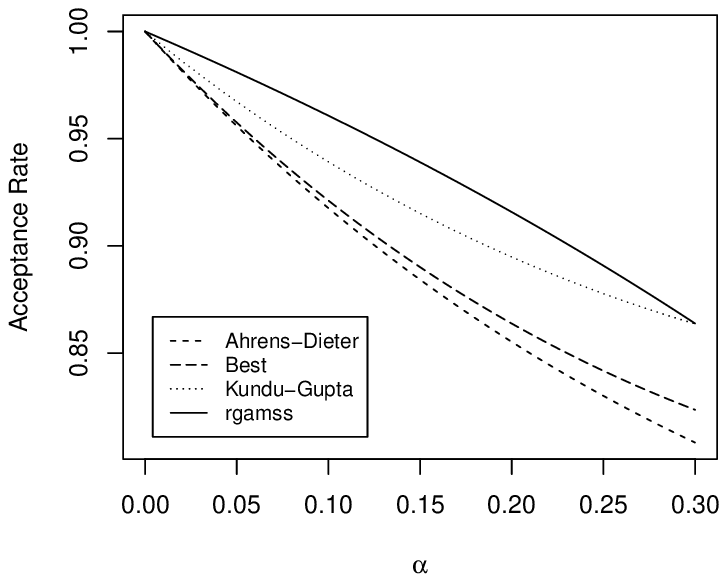}}
\end{center}
\caption{Plot of the acceptance rates for the four indicated methods as a function of the shape parameter $\alpha \in (0, 0.3]$.}
\label{fig:accept}
\end{figure}




\section*{Acknowledgments}

This work is partially supported by the U.S.~National Science Foundation, grants DMS--1007678, DMS--1208833, and DMS--1208841.

\bibliographystyle{apalike}
\bibliography{/Users/rgmartin/Dropbox/Research/mybib}

\end{document}